\documentclass{amsart}

\usepackage{amsmath,amssymb,latexsym,amsthm,graphicx}

\usepackage{mathrsfs}

\newtheorem{theorem}{Theorem}[section]

\newtheorem{lemma}[theorem]{Lemma}

\newtheorem{prop}[theorem]{Proposition}

\newtheorem{defi}[theorem]{Definition}

\newcommand{\reftheo}[1]{Theorem~\ref{#1}}

\newcommand{\refprop}[1]{Proposition~\ref{#1}}

\newcommand{\la}{\left<}

\newcommand{\ra}{\right>}

\newcommand{\supp}{\operatorname{supp}}

\newcommand{\mB}{\mathcal{B}}

\newcommand{\mC}{\mathcal{C}}

\newcommand{\mE}{\mathcal{E}}

\newcommand{\mR}{\mathcal{R}}

\newcommand{\mD}{\mathcal{D}}

\newcommand{\bH}{\mathbb{H}}

\newcommand{\mL}{\mathcal{L}}

\newcommand{\mK}{\mathcal{K}}

\newcommand{\rN}{\mathbb{R}}

\newcommand{\mT}{\mathcal{T}}

\newcommand{\mO}{\mathcal{O}}

\newcommand{\uS}{\mathbb{S}}

\newcommand{\intl}{\int\limits}

\newcommand{\cl}{\operatorname{cl}}

\newcommand{\mV}{\mathcal{V}}

\newcommand{\wf}{\mbox{WF}}

\newcommand{\pdh}{\partial}

\newcommand{\cT}{\mathbb{T}}

\newcommand{\eg}{\varepsilon}

\newcommand{\Llg}{\Lambda}

\newcommand{\ag}{\alpha}

\newcommand{\bg}{\beta}

\newcommand{\ve}{{\bf e}}

\newcommand{\sg}{\sigma}

\newcommand{\ga}{\gamma}

\newcommand{\barni}{\begin{eqnarray*}}
\newcommand{\earni}{\end{eqnarray*}}

\newcommand{\bari}{\begin{eqnarray}}
\newcommand{\eari}{\end{eqnarray}}

\newcommand{\bml}{\begin{multline}}
\newcommand{\eml}{\end{multline}}

\newcommand{\bmlni}{\begin{multline*}}
\newcommand{\emlni}{\end{multline*}}

\title[Artifacts in Generalized X-ray Transform]{On the strength of streak artifacts in limited angle weighted X-ray transform}

\author{Linh V. Nguyen}

\thanks{The research is supported by the NSF grant DMS 1212125}

\address{Department of Mathematics, University of Idaho, Moscow, Idaho 83844, USA}
\email{lnguyen@uidaho.edu}

\begin{document}
\begin{abstract} In this article, we study the limited angle problem for the weighted X-ray transform. We consider the approximate reconstructions by applying two filtered back projection formulas to the limited data. We prove that each resulted operator can be decomposed into the sum of three Fourier integral operators whose symbols are of types $(\varrho,\delta) \neq (1,0)$. The first operator, being a pseudo-differential operator, is responsible for the reconstruction of visible singularities. The other two are responsible for the generation of the artifacts. The theory of Fourier integral operators then implies, in particular, the continuity of the reconstruction operator and geometry of  the artifacts. We then extend the technique developed by the author in [Inverse Problems 31 (2015) 055003] to obtain more refined microlocal estimates for the strength of the artifacts.
\end{abstract}
\maketitle

\section{Introduction}\label{S:intro}
Let us denote by $\uS^1$ the unit circle in $\rN^2$ and $\mu \in C^\infty(\rN^2 \times \uS^1)$ be a strictly positive function. We consider the weighted X-ray (or Radon) transform
$$\mR_{\mu} f(\theta, s) = \int\limits_{x \cdot \theta = s} \mu(x,\theta) \, f(x) \, dx,\quad (\theta,s) \in \uS^1 \times \rN.$$ We are interested in reconstructing $f$ from $\mR_\mu(f)$.

\medskip

When $\mu \equiv 1$, $\mR_\mu$ is the classical X-ray transform $\mR$, which appears in computed X-ray tomography (see, e.g., \cite{Nat}).  The function $f$ can be exactly reconstructed from $\mR(f)$ by the following filtered back-projection inversion formula
\begin{equation} \label{E:Exact} f  =  \mB_0 f := \frac{1}{4 \pi }\mR^* H \frac{d}{ds} \mR f.\end{equation}
Here, $H: \mE'(\rN) \to \mD'(\rN)$ is the Hilbert transform 
$$(H k)(t) = \frac{1}{\pi} ~p.v.~ \intl_{\rN} \frac{k(s)}{t-s} ds,$$
and $\mR^*$ is the formal adjoint of $\mR$, defined by 
$$\mR^*g (x) = \intl_{\uS^1} g(\theta,x \cdot \theta) d\theta.$$
Since $H$ is a nonlocal operator, in order to compute $f$ at any location $x \in \rN^2$, formula (\ref{E:Exact}) needs the data $\mR(\theta,s)$ for all $(\theta,s) \in \uS^1 \times \rN$.

The following local formula, which is called Lambda reconstruction, provides a simple method to reconstruct the singularities of $f$ from $\mR f$
\begin{equation} \label{E:Lambda} \Lambda_0 f :=  \frac{1}{4 \pi} \mR^*(-\pdh_s^2) \mR f.\end{equation}
A main advantage of Lambda reconstruction is that in order to find $\Lambda_0 f(x)$, one only needs the {\bf local} data $\mR(f)(\theta,s)$ in a neighborhood of the set $\{(\theta,s): x \cdot \theta =s \}$. Discussion about Lamda tomography can be found in, e.g., \cite{vainberg1981reconstruction,smith1985mathematical,Fari92,Fari97}. The reader is also referred to \cite{KLM,RamKat-Book} for other kinds of local tomography.
\medskip

There is also a large amount of work dealing with the weighted X-ray transform $\mR_\mu$ \cite{Boman-Quinto, Quinto-80,Quinto-83,Quinto-93-pompiu,Kurusa}. However, there is no explicit exact reconstruction formula for a general function $\mu$. Moreover, the function $f$ may not be uniquely determined from $\mR_\mu f$ \cite{Boman}. On the other hand, local uniqueness holds \cite{Markoe-Quinto}. A special type of the generalized X-ray transform is called attenuated X-ray transform, which appears in  SPECT  (single photon emission computed tomography). The study of attenuated X-ray transform is well established (e.g., \cite{Natterer-79,Natterer-83,Finch-Range-Attenuated,Kazantsev,Novikov-Attenuated,Natterer-01,Novikov-CRA,Nov-Kun,kun-2001-spect,kuchment2014radon}). There are two techniques for inverting the attenuated X-ray transform: the complexification method by \cite{Novikov-Attenuated} and the A-analytic method by \cite{Kazantsev}. The reader is referred to \cite{Finch-Review,Kuchment-review} for comprehensive reviews on the attenuated X-ray transform. 

\medskip

In this article, we are interested in the singularity reconstruction of $f$ from $\mR_\mu f$. Let $\nu \in C^\infty(\rN^2 \times \uS^1)$ be a strictly positive function. We define the following back-projection type operator
$$\mR_\nu^* g(x) = \intl_{\uS^1} g(\theta, x \cdot \theta) \nu(x,\theta) d\theta,$$
and following analogs of $\mB_0$ and $\mL_0$:
\barni \mB f &:=&\frac{1}{4 \pi} \mR_\nu^* H \frac{d}{ds} \mR_\mu f,\\
\Lambda f &:=&  \frac{1}{4 \pi} \mR_\nu^*(-\pdh_s^2) \mR_\mu f.
\earni
Then, $\mB$ and $\Lambda$ are respectively pseudo-differential operators with amplitudes  (see, e.g., \cite{KLM,katsevich1998local}): $$a_{\mB}(x,y,\xi) = \frac{1}{2} \left[\nu(x,\xi) \, \mu(y,\xi) + \nu(x,-\xi) \, \mu(y,-\xi)\right],$$
and $$a_\Llg(x,y,\xi) = \frac{1}{2}|\xi| \, \left[ \nu(x,\xi) \, \mu(y,\xi) +  \nu(x,-\xi) \, \mu(y,-\xi)\right].$$
Here, we have extended $\mu,\nu$ to positively homogeneous functions of degree zero with respect to $\xi$. Since $a_\mB$ is homogenous of degree zero and non-vanishing, $\mB f$ reconstructs all the singularities of $f$ with the exact order. Since $a_\Llg$ is homogenous of degree one and non-vanishing, $\Llg f$ emphasizes all the singularities of $f$ by one order (see more details in \cite{KLM,kuchment2014radon}). 

\medskip

We now consider the limited angle problem (see, e.g., \cite{RamKat-AML1,RamKat-AML2,KLM, Kat-JMAA,FQ13,FQ-Para,kuchment2014radon}): 

\medskip

\begin{center} {\bf $\mR_\mu f(\theta,s)$ is only known for $\theta \in \uS_\mV :=\{(\cos \phi, \sin \phi): \phi_1< \phi< \phi_2\}$, where $0<\phi_1<\phi_2< \pi$. } \end{center}

\medskip

Let us consider  $\kappa \in C^\infty(\overline{\uS}_{\mV})$ such that $\kappa>0$ on $\uS_\mV$.  We then extend $\kappa$ to $\uS^1$ by zero and define the limited angle version of $\mB$ and $\Llg$: 
\barni \mB_{\kappa} f &:=&\frac{1}{4 \pi} \mR_{\nu}^* \, \kappa  \, H \, \frac{d}{ds} \mR_\mu f,\\
\Llg_{\kappa} f &:=&  \frac{1}{4 \pi} \mR_{\nu}^* \, \kappa \, (-\pdh_s^2) \mR_\mu f.
\earni
Then, one can write $\mB_{\kappa}$ and $\Llg_{\kappa}$ as oscillatory integrals:
\begin{equation*}  \mB_\kappa f(x) = \frac{1}{(2 \pi)^2}\intl_{\rN^2} e^{i (x-y) \cdot \xi }  a_{\mB,\kappa}(x,y,\xi) \, f(y) \, d \xi \,dy,\end{equation*}
and
\begin{equation*}  \Llg_\kappa f(x) = \frac{1}{(2 \pi)^2}\intl_{\rN^2} e^{i (x-y) \cdot \xi }  a_{\Llg,\kappa}(x,y,\xi) \, f(y) \, d \xi \,dy. \end{equation*}
Here,
$$a_{\mB,\kappa} (x,y,\xi) = \frac{1}{2} \left[\nu(x,\xi) \, \mu(y,\xi) \, \kappa (\xi/|\xi|) \,   + \nu(x,- \xi) \, \mu(y,-\xi) \, \kappa (-\xi/|\xi|) \right],$$
and $$a_{\Llg,\kappa} (x,y,\xi) =\frac{1}{2} |\xi| \, \left[ \nu(x,\xi) \, \mu(y,\xi) \,\kappa(\xi/|\xi|)+  \nu(x,-\xi) \, \mu(y,-\xi) \,\kappa(-\xi/|\xi|) \right].$$

\medskip

Assume that $\kappa$ vanishes to infinite order at the end points of $\uS_\mV$, then $\mB_{\kappa}$ and $\Llg_{\kappa}$ are pseudo-differential operators with principal symbols respectively: $$\sg_{\mB,\kappa} (x,\xi) = \frac{1}{2} \left[\kappa (\xi/|\xi|) \, \nu(x,\xi) \, \mu(x,\xi) +\kappa (-\xi/|\xi|) \, \nu(x,-\xi) \, \mu(x,-\xi) \right],$$
and $$\sg_{\Llg,\kappa} (x,\xi) = \frac{1}{2} |\xi| \left[\kappa (\xi/|\xi|) \, \nu(x,\xi) \, \mu(x,\xi) +\kappa (-\xi/|\xi|) \, \nu(x,-\xi) \, \mu(x,-\xi) \right].$$
We note that $\sg_{\mB,\kappa} (x,\xi)>0$ and $\sg_{\Llg,\kappa} (x,\xi)>0$ if $\frac{\xi}{|\xi|} \in \uS_\mV$ or $\frac{-\xi}{|\xi|} \in \uS_\mV$. Therefore, $\Llg_{\kappa}$ and $\mB_{\kappa}$ reconstruct all the {\bf visible} singularities, with the same order as $\Llg$ and $\mB$ do, respectively. We recall that visible singularities are all $(x,\xi) \in \wf(f)$ such that $\frac{\xi}{|\xi|} \in \uS_\mV$ or $\frac{-\xi}{|\xi|} \in \uS_\mV$ (see, e.g., \cite{KLM}). Meanwhile, they do not generate any added singularities (i.e., artifacts). The reader is referred to \cite{KLM} for detailed discussion.

\medskip

Let us now consider the case when $\kappa$ only vanishes to a finite order $k$ at the end points of $\uS_\mV$  \footnote{That is, $\kappa^{(l)} (\theta(\phi))=0$, $0 \leq l \leq k-1$ and $\kappa^{(k)} (\theta(\phi))\neq 0$, for $\phi=\phi_1$ and $\phi=\phi_2 $.}. We denote $\ve_1=(\cos \phi_1, \sin \phi_1), \quad \ve_2=(\cos \phi_2, \sin \phi_2)$. Then, the amplitudes $a_{\mB,\kappa}$ and $a_{\Llg,\kappa}$ are not smooth with respect to $\xi$, across the lines
$$\ell_1= \{\xi \in \rN^2: \xi = r \ve_1: r  \in \rN \}$$ and $$\ell_2= \{\xi \in \rN^2: \xi = r \ve_2: r  \in \rN \}.$$
Therefore, $\mB_{\kappa}$ and $\Llg_{\kappa}$
are no longer pseudo-differential operators in the standard sense. They are, instead, pseudo-differential operators with singular symbols.  Moreover, it is observed that  $\mB_{\kappa} f$ and $\Llg_{\kappa}f$ may contain artifacts (see, e.g. \cite{Kat-JMAA,Frikel-Quinto-2015}). {\bf The goal of this article is to study $\mB_{\kappa}$ and $\Llg_{\kappa}$ in this case}. 

In \cite{Kat-JMAA}, Katsevich obtains the geometric characterization of the artifacts. He also analyzes the strength of the artifacts for the case $f$ only has jump singularities. His approach relies on the direct estimates for oscillatory integrals.  Recently, Frikel and Quinto \cite{FQ-Para,Frikel-Quinto-2015} provide a general paradigm to study the geometric description for artifacts arising in limited angle problem of the weighted X-ray and Radon transforms. Their approach relies on the calculus of wave front set for compositions.

This article, although having some overlaps with the above mentioned works, has some distinct features. Firstly, we prove a decomposition of $\mT$ as a sum of three Fourier operators with non-classical symbols. This result, being interesting in itself, implies the geometric description of the artifacts. Moreover, it provides the continuity of $\mT$ between Sobolev spaces. Such a continuity has not been obtained before. Secondly, we analyze the strength of the artifacts, when $f$ is an {\bf arbitrary} compactly supported distribution. Our approach for this goal is a refinement of the previous work \cite{Streak-Artifacts}. 

It is worth mentioning that some similar analysis of artifacts for spherical Radon transform has been done in several works \cite{FQ15,SIMA,BFNguyen}. 



To proceed, we will consider a family of operators which contains $\mB_{\kappa}$ and $\Llg_{\kappa}$ as special cases. Let $\mT$ be the linear operator whose Schwartz kernel is:
\begin{equation} \label{E:sch} \mK(x,y) = \frac{1}{(2 \pi)^2}\intl_{\rN^2} e^{i (x-y) \cdot \xi }  a(x,y,\xi) \, \chi(\xi) \, d \xi,\end{equation} where $a \in S^m((\rN^2 \times \rN^2) \times \rN^2)$ \footnote{See the definition of the symbol class $S^m$  in Definition~\ref{D:Sym} and the following discussion. }.
Here, $\chi$ is the characteristic function of $\cl(\mathbb{W})$, where 
$$\mathbb{W} = \rN \, \uS_\mV = \{\xi = r (\cos \theta, \sin \theta): \phi_1 < \theta < \phi_2, \, r \in \rN\}.$$
 Obviously, $\mB_{\kappa}$ and $\Llg_{\kappa}$ are special cases of $\mT$ (with $m=0$ and $m=1$, respectively). In the sequel, we analyze $\mT$, and then interpret our results to $\mB_{\kappa}$ and $\Llg_{\kappa}$.

\medskip

Let $ \Delta \subset (\mathbb{T}^*\rN^{2} \setminus 0) \times (\mathbb{T}^*\rN^{2} \setminus 0)$ be the diagonal relation $$\Delta = \{(x,\xi; x, \xi): (x,\xi) \in \cT^*\rN^{2} \setminus 0\},$$
and $$\Delta_0=\{(x,\xi; x, \xi)\in \Delta: \xi \in \cl(\mathbb{W})\}.$$
For $j=1,2$, we define $ \mC_j \subset (\mathbb{T}^*\rN^{2} \setminus 0) \times (\mathbb{T}^*\rN^{2} \setminus 0)$ by
\begin{eqnarray*}
\mC_j = \{(x,\ga\, \ve_j; x + t \ve_j^\perp, \ga\, \ve_j): x \in \rN^2,~\ga,t \in \rN,~\ga \neq 0\}.
\end{eqnarray*}
Let us recall the definition of the symbol classes $S^m_{\varrho,\delta}(X \times \rN^N)$ (see, e.g., \cite{hormander71fourier,Shubin-book}):
\begin{defi} \label{D:Sym}
Let $m, \varrho$ and $\delta$ be real numbers, $0 \leq \varrho, \delta \leq 1$. The class $S_{\varrho,\delta}^m(X \times \rN^N)$ consists of all functions $\rho(x,\xi) \in C^\infty(X \times (\rN^N \setminus 0))$ such that for any multi-indices $\ag, \bg$ and any compact set $K \Subset X$ there exists a constant $C=C_{\ag,\bg,K}$ for which
$$|\pdh_\xi^\ag \pdh_x^\bg \rho(x,\xi)| \leq C (1+ |\xi|)^{m- \varrho |\ag| + \delta |\bg|},$$ for all $(x,\xi) \in X \times \rN^N$. 
\end{defi}
An element $\rho \in S_{\varrho,\delta}^m(X \times \rN^N)$ is called a symbol of order $m$ and type $(\varrho,\delta)$. We also denote $S^m(X \times \rN^N) =S^m_{1, 0}(X \times \rN^N)$, which are the most common type of symbol classes (especially in tomography). In this article, however, we need to make use of symbol classes $S_{\varrho,\delta}^m(X \times \rN^N)$ where $(\varrho,\delta) \neq (1, 0)$. Let $\mC$ be a Lagrangian in the cotangent bundle $\cT^* X$ of $X$. We will denote by $I^m_{\varrho,\delta}(\mC)$ the class of Fourier integral distributions of order $m$ whose symbol is of type $(\varrho,\delta)$ and canonical relation is a subset of $\mC$. The order of a Fourier integral distribution is not necessarily the same as the order of its symbol (we will elaborate on this fact later when needed). The reader is referred to \cite{hormander71fourier,Shubin-book} for the detailed treatment on Fourier integral distributions. 

\medskip

Here is the main result of this article:
\begin{theorem}\label{T:Main1}
We have
\begin{itemize}
\item[a)] For any $n>0$, we can write $$\mK = \mK_{0,n} + \mK_{1,n} + \mK_{2,n},$$ where $$\mK_{0,n} \in I^m_{\frac{1}{n}, 0}(\Delta), \quad \mK_{1,n} \in I^{m+\frac{1}{n}-\frac{1}{2}}_{1,\frac{1}{n}}(\mC_1), \quad\mK_{2,n} \in I^{m+\frac{1}{n}-\frac{1}{2}}_{1,\frac{1}{n}}(\mC_2).$$ Moreover, the symbol $\sg(x,\xi)$ of $\mK_{0,n}$ satisfies 
\begin{itemize}
\item[i)] $\sg(x,\xi) - a(x,x,\xi)  \in S_{\frac{1}{n},0}^{m-\frac{1}{n}}(\rN^2 \times V),$ for any closed conic set $V \subset \mathbb{W}$; and 
\item[ii)] $\sg(x,\xi) \in S_{\frac{1}{n},0}^{-\infty}(\rN^2 \times V),$ for any closed conic set $V \subset \rN^2 \setminus \cl(\mathbb{W})$.
\end{itemize}
\medskip

\item[b)]  Assume that $a(x,y,\xi)$ vanishes to order $k$ across the line $\ell_i$. Then, near $\mC_j \setminus \Delta$, $\mK $ is microlocally in the space $I^{m-k-1/2}(\mC_j)$. 
\end{itemize}
\end{theorem}

Let $\mT_{0,n}$, $\mT_{1,n}$, and $\mT_{2,n}$ be the operators whose Schwartz kernels are $\mK_{0,n}$, $\mK_{1,n}$, and $\mK_{2,n}$, respectively. We obtain from Theorem~\ref{T:Main1}~a) the following continuity:
\begin{itemize}
\item[(C.1)] Since $\mK_{0,n} \in I^m_{\frac{1}{n}, 0}(\Delta)$, $\mT_{0,n}$ is a continuous map from $H^{s}_{comp} (\rN^2)$ to $H^{s-m}_{loc}(\rN^2)$ (see, e.g., \cite[Theorem~7.1]{Shubin-book}), and 
\item[(C.2)] Since $ \mK_{1,n} \in I^{m+\frac{1}{n}-\frac{1}{2}}_{1,\frac{1}{n}}(\mC_1)$, and $\mK_{2,n} \in I^{m+\frac{1}{n}-\frac{1}{2}}_{1,\frac{1}{n}}(\mC_2)$, $\mT_{1,n}$ and $\mT_{2,n}$ are continuous maps from $H^s_{comp}(\rN^2)$ to $H_{loc}^{s-m-\frac{1}{n}}(\rN^2)$ (see, e.g., \cite[Theorem 4.3.2]{hormander71fourier}). It should be noticed here the difference between the order $m+\frac{1}{n}-\frac{1}{2}$ of $\mK_{1,n}, \mK_{2,n}$ and the aforementioned mapping property. This comes from the fact that the canonical relations $\mC_1$ and $\mC_2$ are not local graph. Their left and right projections are fibered of dimension one. 
\end{itemize}

(C.1) and (C.2), in particular, imply the continuity of $\mT$ from $H^s_{comp}(\rN^2)$ to $H_{loc}^{s-m-\frac{1}{n}}(\rN^2)$, for any $n>0$. To the best of our knowledge, this continuity has not been proved anywhere. We note here that if $a(x,y,\xi)$ vanishes to infinite order at the boundary of $\mathbb{W}$, then   from the standard theory of pseudo-differential operator, $\mT$ is a continuous map from $H^s_{comp}(\rN^2)$ to $H_{loc}^{s-m}(\rN^2)$. It is interesting to see whether such optimal bound also holds for the case $a(x,y,\xi)$ only vanishes to finite order, as  being considered in this article.

Let us interpret \reftheo{T:Main1} further as properties of $\mB_\kappa$ and $\Llg_\kappa$ (where $m=0$ and $m=1$, respectively). From \reftheo{T:Main1}~a), we obtain $$\wf(\mK) \subset \Delta_0 \cup \mC_1 \cup \mC_2. $$
This result was obtained in \cite{Kat-JMAA,FQ-Para} by other methods. It, in particular, describes the geometry of the artifacts introduced in $\mB_{\kappa} f$ and $\Llg_{\kappa} f$ \footnote{An artifact in $\mT f$ is a singularity $(x,\xi) \in \wf(\mT f)$ such that $(x,\xi) \not \in \wf(f)$.}. The artifacts are generated by the ``edge" singularities via the canonical relations $\mC_1$ and $\mC_2$, as follows. An ``edge" singularity is an element $(x,\xi) \in \wf(f)$ such that $\xi \parallel \ve_j$ ($j=1,2$). It may generate the artifacts at other elements $(y,\xi) \neq (x,\xi)$ satisfying $(y,\xi;x,\xi) \in \mC_j$. That is, the artifacts generated by $(x,\xi)$ are located along the line passing through $x$ and orthogonal to $\ve_j$.  The same phenomenon in the limited angle problem of the standard X-ray transform was presented in \cite{FQ13,Streak-Artifacts}. 

From (C.1), we obtain that the reconstructed singularities are at most $m$ order(s) stronger than the original ones. Moreover, let us recall that the amplitudes $a_{\mB,\kappa}$ and $a_{\Llg,\kappa}$ (of $\mB_\kappa$ and $\Llg_\kappa$, respectively) are nonvanishing on $\mathbb{W}$. The formula for the symbol of $\mK_{0,n}$ in Theorem~\ref{T:Main1}~a) shows that the singularity $(x,\xi) \in \wf(f) $ satisfying $\xi \in \mathbb{W}$ (i.e., $(x,\xi)$ is a {\bf visible} singularity) is reconstructed and the reconstructed singularity is exactly $m$ order stronger than the original singularity. On the other hand, the singularity $(x,\xi) \in \wf(f)$ such that $\xi \not \in \overline{\mathbb{W}}$ (i.e., $(x,\xi)$ is an {\bf invisible} singularity) is completely smoothened out by $\mB_{\kappa}$ and $\Llg_{\kappa}$. This is due to the fact that the symbol of $\mK_{0,n}$ vanishes for $\xi \not \in \overline{\mathbb{W}}$ (see Theorem~\ref{T:Main1} a) ii)). Let us mention that the descriptions presented in this paragraph have been obtained previously in several works \cite{Kat-JMAA,FQ13,Streak-Artifacts} by other methods. 



\reftheo{T:Main1}~b) is a generalization of \cite[Theorem 3.1~b)]{Streak-Artifacts}, where the standard X-ray transform was considered. It provides explicit bounds for the artifacts, as follow (see \cite{Streak-Artifacts} for the detailed explanation): 
\begin{itemize}
\item The artifacts are {\bf at most} $(m-k)$ order(s) stronger than their strongest generating singularities if $m>k$. The artifacts are {\bf at most} as strong as their strongest generating singularities if $m=k$. The artifacts are {\bf at least} $(k-m)$ order(s) smoother than their strongest generating singularities if $k>m$.

\item Assume that the artifact $(x,\xi) \in \wf(\mT f)$ has finitely many generating singularities $(y,\xi) \in \wf(f)$, each of them is conormal to a curve $S$ having non-vanishing curvature at $y$. Then, the artifact is {\bf at most} $(m-k-\frac{1}{2})$ order(s) stronger than its strongest generating singularity if $m>k- \frac{1}{2}$. It is {\bf at most} as strong as its strongest generating singularity if $m=k-\frac{1}{2}$.  It is {\bf at least} $(k-m-\frac{1}{2})$ order(s) smoother than its strongest generating singularity if $k>m-\frac{1}{2}$. 
\end{itemize}
These bounds for the artifacts are better than what can be obtained from (C.2). However, it should be noted Theorem~\ref{T:Main1}~b) does not explain the strength of $(x,\xi) \in \wf(\mT f)$ if $(x,\xi)$ is an ``edge" singularity of $f$ (i.e., $(x,\xi) \in \wf(f)$ and $\xi \parallel \ve_j$ for $j=1$ or $j=2$). Meanwhile, (C.2) implies that such singularity $(x,\xi) \in \wf(\mT f)$ is at most $m+\frac{1}{n}$ order(s) stronger than $(x,\xi) \in \wf(f)$, for any $n$. 




Let us briefly discuss the main ideas for the proof of \reftheo{T:Main1}. The proof of \reftheo{T:Main1}~a) uses a nonlinear cutoff for the phase variable $\xi$. It is a modification of the well known parabolic cutoff, first used by Boutet de Monvel \cite{Monvel} to deal with hypoelliptic operators, and later by Guillemin to deal with the pseudo-differential operators with singular symbol (see the discussion in \cite{GrU-Functional}). The proof of \reftheo{T:Main1}~b) is similar to that of \cite[Theorem 3.1]{Streak-Artifacts}. Namely, it makes use of some proper integrations by parts. However, it is worth mentioning that our argument in this article is cleaner than that in \cite{Streak-Artifacts}, where the special case - standard X-ray transform - is studied.

A deep theory of pseudo-differential operators with singular symbols was developed by Guillemin, Melrose, Uhlmann and others (see, e.g., \cite{MU,GU,AnUhl}). The use of that theory to analyze the X-ray transform, when the canonical relation is not a local canonical graph, was pioneered by  Greenleaf and Uhlmann \cite{GrUDuke,GrUCon}. It has been then exploited intensively to analyze other imaging scenarios (e.g., \cite{FLU,NoChe,FeleaCPDE,FeleaQuinto,Suresh,FeleaSAR,Am-singular}). Although we do not make use of it explicitly, our analysis is inspired by that theory.

In the next section, we will present the proof of \reftheo{T:Main1}. We will start by studying a model distribution defined by an oscillatory integral, whose amplitude is in the class $S^m$ when $\xi$ is away from a straight line through the origin. We show that it can be decomposed into two parts. One belongs to $I^m_{\frac{1}{n},0}(\Delta)$ and the other to $I^{m+\frac{1}{n}-\frac{1}{2}}_{1,\frac{1}{n}}(\mC)$, where $\mC$ is a Lagrangian defined later. 
Moreover, the model distribution is microlocally in the space $I^{m-k-\frac{1}{2}}(\mC)$ near $\mC \setminus \Delta$. With all these results in hands, we then prove \reftheo{T:Main1} by a simple partition of unity argument.

\section{Proof of \reftheo{T:Main1}}
\subsection{Model oscillatory integrals}\label{S:Model}
\noindent Let $a(x,y,\xi) \in S^m((\rN^2 \times \rN^2) \times \rN^2)$ such that for any $(x,y) \in \rN^2 \times \rN^2$ there is $M_{x,y}>0$: \begin{equation} \label{E:sup} \supp a(x,y,.) \subset\{\xi: |\xi_2| \leq M_{x,y} \, |\xi_1|\}.\end{equation} We consider the oscillatory integral 
\bari \label{E:mupm} \mK_\pm(x,y) =\frac{1}{(2\pi)^2} \intl_{\rN^2} e^{i (x-y) \cdot \xi } \, a(x,y,\xi) \, \bH(\pm \xi_2) \, d \xi. \eari
Here, $\bH$ is the Heaviside function, defined by
\begin{eqnarray*}
\bH(s) = \left\{\begin{array}{l}1,\quad s \geq 0, \\[3 pt] 0, \quad s<0. \end{array} \right.
\end{eqnarray*}
We also recall the diagonal canonical relation in $(\cT^* \rN^2 \setminus 0) \times (\cT^* \rN^2 \setminus 0)$ $$\Delta = \{(x,\xi; x, \xi): (x,\xi) \in \cT^*\rN^{2} \setminus 0\},$$
and define $\mC \subset (\cT^* \rN^2 \setminus 0) \times (\cT^* \rN^2 \setminus 0)$ by
$$\mC = \{(x,\xi; y, \xi) \in (\cT^* \rN^2 \setminus 0) \times (\cT^* \rN^2 \setminus 0): x_1-y_1=0,\xi_2=0\}.$$

\medskip
\noindent The following results characterize the distribution $\mK_\pm$:

\begin{prop} \label{P:decompose} We have 
\begin{itemize}
\item[a)] For any $n \geq 2$, we can write $\mK_\pm= \mK_1 + \mK_2$
where $\mK_1 \in I^m_{\frac{1}{n},0} (\Delta)$ and $\mK_2 \in  I^{m+ \frac{1}{n} - \frac{1}{2}}_{1,\frac{1}{n}}(\mC)$. Moreover, the symbol $\sg(x,\xi)$ of $\mK_1$ satisfies: 
\begin{itemize}
\item[i)] $\sg(x,\xi) - a(x,x,\xi) \in S_{\frac{1}{n},0}^{m-\frac{1}{n}}(\rN^2 \times V)$ for any conic closed conic set $V \subset \{ \xi \in \rN^2: \pm \xi_2 > 0\}$, and 
\item[ii)] $\sg(x,\xi) \in S_{\frac{1}{n},0}^{-\infty}(\rN^2 \times V)$ for any conic closed conic set $V \subset \{ \xi \in \rN^2: \pm \xi_2 < 0\}$. 
\end{itemize}
\item[b)] Assume that $a(x,y,\xi)$ vanishes up to order $k$ at $\xi_2 = 0$. Then, $\mK_\pm$ is microlocally in the space  $I^{m-k-\frac{1}{2}}(\mC)$ near $\mC \setminus \Delta$.
\end{itemize}

\end{prop}

\begin{proof}  We only consider $\mK_+$, since the proof for $\mK_-$ is similar. 

\medskip

\noindent We first prove a). Let $c \in C^\infty(\rN)$ be a cut-off function satisfying: $c(\tau)=1$ for $|\tau| \leq 1$ and $c(\tau) =0$ for $|\tau| \geq 2$. Let us define \footnote{Here, we use the Japanese bracket $\left< . \right>$ convention: $\left<\xi \right> = (1+|\xi|^2)^{1/2}$.}
\barni \mK_1(x,y) =\frac{1}{(2\pi)^2} \intl_{\rN^2} e^{i (x-y) \cdot \xi } \,\Big[1- c\big(\xi_2^n/\left<\xi\right>\big)\Big] \, a(x,y,\xi) \, \bH(\xi_2) \, d \xi,
\earni
and 
\barni \mK_2(x,y) =\frac{1}{(2\pi)^2} \intl_{\rN^2} e^{i (x-y) \cdot \xi } \,c\big(\xi_2^n/ \la \xi \ra \big)\, a(x,y,\xi) \, \bH(\xi_2) \, d \xi. \earni
Then, $\mK_+ =\mK_1 + \mK_2$. We now prove that $\mK_1$ and $\mK_2$ satisfy the properties stated in \refprop{P:decompose}~a).

\medskip

To analyze $\mK_1$, let us denote $$b_1(x,y,\xi) = \Big[1- c\big(\xi_2^n/\left<\xi\right>\big)\Big] \, a(x,y,\xi) \, \bH(\xi_2).$$
Then, $b_1 \in C^\infty((\rN^2 \times \rN^2) \times (\rN^2 \setminus 0))$. Moreover, for any multi-index $\ag=(\ag_1,\ag_2)$:
$$\pdh_\xi^{\ag} b_1(x,y,\xi) = \sum_{l=0}^{\ag_2} \sum_{k=0}^{\ag_1}  \pdh_{\xi_2}^{l} \pdh_{\xi_1}^{k} \left[1-c \big(\xi^n_2/\left<\xi\right> \big)\right] \, \pdh_{\xi_2}^{\ag_2-l}   \pdh_{\xi_1}^{\ag_1-k}  a(x,y,\xi) \, \bH(\xi_2). $$
Since $c(\tau)=1$ for $|\tau| \leq 1$ and $c(\tau)=0$ for $|\tau| \geq 2$, 
$$\big| \pdh_{\xi_2}^{l} \pdh_{\xi_1}^{k}  \left[1-c \big(\xi^n_2/\left<\xi\right> \big)\right]\big| \leq C \left<\xi_1 \right>^{-\frac{l}{n} -k}.$$
Noting that $a \in S^m((\rN^2 \times \rN^2) \times \rN^2)$, we obtain
$$|\pdh_\xi^{\ag} b_1(x,y,\xi)| \leq C \left< \xi \right>^{m-\frac{|\ag|}{n}}.$$
It is, hence, straight forward to show that $b_1 \in  S^m_{\frac{1}{n},0}((\rN^2 \times \rN^2) \times \rN^2)$. Therefore,
$\mK_1 \in I^m_{\frac{1}{n},0}(\Delta)$.

Let $V \subset \{\xi \in \rN^2: \xi_2 > 0\}$ be a closed conic set. Then, there is $\eg>0$ such that $|\xi_2| \geq \eg |\xi_1|$ for all $\xi \in V$.  Hence, for big enough $R>0$, 
$$|\xi_2|^n \geq 2 |\xi_1|,\quad \mbox{ if } \xi \in V \mbox{ and } |\xi| \geq R.$$
That is,
$$b_1(x,y,\xi) = a(x,y,\xi),  \quad \mbox{ of } \xi \in V \mbox{ and } |\xi| \geq R.$$
Thus, $$b_1 - a \in S_{\frac{1}{n},0}^{-\infty}(\rN^2 \times V).$$
On the other hand, using the formula of symbol (e..g, \cite[(2.1.4)]{hormander71fourier}), we obtain: $$\sg(x,\xi) - b_1(x,x,\xi) \in S_{\frac{1}{n},0}^{m-\frac{1}{n}}(\rN^2 \times \rN^2).$$ Therefore,
$$\sg(x,\xi) - a(x,x,\xi) \in S_{\frac{1}{n},0}^{m-\frac{1}{n}}(\rN^2 \times V).$$
This proves i).
On the other hand, assume that $V \subset \{\xi \in \rN^2: \xi_2 < 0\}$ is a closed conic set. Then 
$$b_1(x,y,\xi) =0,\quad (x,\xi) \in \rN^2 \times V.$$
This implies $$\sg(x,\xi)  \in S_{\frac{1}{n},0}^{-\infty}(\rN^2 \times V).$$
We, therefore, have proved ii) and finished the analysis of $\mK_1$.

\medskip

We now analyze $\mK_2$. Let us write
\barni \mK_2(x,y) =\frac{1}{(2\pi)^2} \intl_{\rN} e^{i (x_1-y_1) \, \xi_1 } b_2(x,y,\xi_1) \, d \xi_1,\earni
where
\barni b_2(x,y,\xi_1) =\frac{1}{(2\pi)^2} \intl_{\rN} e^{i (x_2-y_2) \xi_2 } \,c\big(\xi^n_2/\la \xi \ra \big) \, a(x,y,\xi) \, \bH(\pm \xi_2) \, d \xi_2. \earni
Since $c(\tau)=0$ for $|\tau| \geq 2$, there is $M= M(n)$ such that for all $\xi$ satisfying $c\big(\xi^n_2/\la \xi \ra \big) \neq 0$: $$|\xi_2| \leq M \la \xi_1 \ra^{\frac{1}{n}}.$$ 
Therefore, since $a \in S^m((\rN^2 \times \rN^2) \times \rN^2)$, for any compact set $K \Subset \rN^2 \times \rN^2$:
\barni |b_2(x,y,\xi_1)| \leq C\intl_0^{M \la \xi_1 \ra^{\frac{1}{n}}}  \la \xi \ra^{m}  d \xi_2 \leq C \la \xi_1 \ra^{m+ \frac{1}{n}}, \quad \forall (x,y) \in K. \earni
Similarly, we obtain
\barni |\partial_x^\ag \partial_y^\bg \partial_{\xi_1}^l \, b_2(x,y,\xi_1)| \leq C \la \xi_1 \ra^{m+ \frac{1}{n}+ \ag_2 \frac{1}{n}+ \bg_2 \frac{1}{n} - l}. \earni
Therefore, $$b_2 \in S_{1,\frac{1}{n}}^{m+ \frac{1}{n}}((\rN^2 \times \rN^2) \times \rN^2).$$We arrive to $\mK_2 \in S^{m + \frac{1}{n}- \frac{1}{2}}_{1,\frac{1}{n}}(\mC)$. We note here that the difference between the orders of $b_2$ and $\mK_2$ comes the following formula (see, e.g., \cite{hormander71fourier}):
 \begin{equation} \label{E:Order} \mbox{ order of } \mK_2= \mbox{ order of } b_2 + (N-n)/2, \end{equation} where $n=2$ is the dimension of the spatial variables $x,y$ and $N=1$ is the dimension of the phase variable $\xi_1$. The proof of  a) is finished. 

\medskip

Let us now prove b). Let $(x^*, \xi^*;  y^*, \xi^*) \in \mC \setminus \Delta$, then $x^*_2 \neq y^*_2$. Let $\mO \subset \rN^2 \times \rN^2$ be an open set containing $(x^*,y^*)$ such that $x_2 \neq y_2$ for any $(x,y) \in \mO$. It suffices to prove that $\mK_+|_{\mO}  \in I^{m-k-\frac{1}{2}}(\mC)$. Let us write
\begin{equation} \label{E:muc} \mK_+(x,y) = \frac{1}{(2 \pi)^2}  \intl_{\rN} e^{i (x_1-y_1) \xi_1 }  \, a_1(x,y,\xi_1) \,d \xi_1,\end{equation} where
$$a_1(x,y,\xi_1) = \intl_{\rN} e^{i (x_2-y_2) \xi_2 } \, a(x,y,\xi) \, \bH(\xi_2) \, d\xi_2 = \intl_{0}^\infty e^{i (x_2-y_2) \xi_2 } \, a(x,y,\xi)  \, d\xi_2.$$
Let us note that, due to (\ref{E:sup}), the above integral is, in fact, over a finite interval. Taking integration by parts $(k+1)$ times with respect to $\xi_2$, we obtain:
\barni a_1(x,y,\xi_1) = \frac{1}{[i(y_2 - x_2)]^{k+1}} \Big[\pdh_{\xi_2}^k a(x,y,\xi_1,0) +  \intl_{0}^\infty e^{i (x_2-y_2) \xi_2 } \,\pdh_{\xi_2}^{k+1}  a(x,y,\xi)  \, d\xi_2 \Big].
\earni
Since $a \in S^m(\mO \times \rN^2)$ the first term on the right hand side belongs $S^{m-k}(\mO \times \rN)$. Due to Lemma \ref{L:rN} below, the second term belongs to $S^{m-k}(\mO \times \rN)$\footnote{Using integration by parts once and apply Lemma \ref{L:rN}, one can show that this term is, in fact,  in $S^{m-k-1}(\mO \times \rN)$.}. Therefore, we arrive to $a_1 \in S^{m-k}(\mO \times \rN)$. Therefore, microlocal near $\mC \setminus \Delta$, $\mK_+ \in I^{m-k - \frac{1}{2}}(\mC)$. Here, we have used a formula similar to (\ref{E:Order}) to determine the order of $\mK_+$. This finishes the proof for b).
\end{proof}

\medskip

\begin{lemma} \label{L:rN} Assume that $\rho(x,y,\xi)  \in S^m((\rN^2 \times \rN^2) \times \rN)$, and for any $(x,y) \in \rN^2 \times \rN^2$ there is $M_{x,y}>0$: $$\supp \rho(x,y,.) \subset\{\xi: |\xi_2| \leq M_{x,y} \, |\xi_1|\}.$$ Let $\mO \subset \rN^2 \times \rN^2$ such that $x_2 \neq y_2$ for all $(x,y) \in \mO$. Define the function 
\barni R(x,y,\xi_1) = \intl_{0}^\infty e^{i (x_2-y_2) \xi_2 } \, \rho(x,y,\xi)  \, d\xi_2.
\earni
Then, $R(x,y,\xi_1) \in S^{m+1}(\mO \times \rN)$. 
\end{lemma}
\begin{proof} We need to prove that for any multi-indices $\ag,\bg$, any integer $l \geq 0$, and compact set $K \Subset \mO$, there is $C= C_{K,\ag,\bg,l}>0$ such that:
\begin{equation} \label{E:rN} |\pdh_x^\ag \pdh_y^{\bg} \pdh_{\xi_1}^l R(x,y,\xi_1)| \leq C \,\la \xi_1 \ra^{m-l+1}, \quad \mbox{ for all } (x,y) \in K.\end{equation}
 We observe that $\pdh_x^\ag \pdh_y^{\bg} \pdh_{\xi_1}^l R(x,y,\xi_1) $ is equal to
\begin{multline*} \sum_{k_1=0}^{\ag_2} \sum_{k_2=0}^{\bg_2}  \int\limits_0^\infty  e^{i \, (x_2-y_2) \, \xi_2 } \,(i \xi_2)^{k_1} \,  (-i \xi_2)^{k_2}  \partial_{x_2}^{\ag_2-k_1}  \partial_{y_2}^{\bg_2-k_2} \partial_{x_1}^{\ag_1}  \partial_{y_1}^{\bg_1} \partial_{\xi_1}^l \rho(x,y,\xi) \, d \, \xi_2 \\=  \sum_{k_1=0}^{\ag_2} \sum_{k_2=0}^{\bg_2}  I_{k_1,k_2}(x,y,\xi_1).
\end{multline*}
Taking integration by parts with respect to $\xi_2$, we obtain
\begin{multline*} I_{k_1,k_2}(x,y,\xi_1) = \int\limits_0^\infty  e^{i \, (x_2-y_2) \, \xi_2 } \,(i \xi_2)^{k_1} \,  (-i \xi_2)^{k_2}  \partial_{x_2}^{\ag_2-k_1}  \partial_{y_2}^{\bg_2-k_2} \partial_{x_1}^{\ag_1}  \partial_{y_1}^{\bg_1} \partial_{\xi_1}^l \rho(x,y,\xi) \, d \, \xi_2 \\ = \frac{1}{[i(x_2-y_2)]^{k_1+k_2} } \int\limits_0^\infty  \partial_{\xi_2}^{k_1+k_2} \, [e^{i \, (x_2-y_2) \, \xi_2 }] \,(i \xi_2)^{k_1} \,  (-i \xi_2)^{k_2}  \partial_{x_2}^{\ag_2-k_1}  \partial_{y_2}^{\bg_2-k_2} \partial_{x_1}^{\ag_1}  \partial_{y_1}^{\bg_1} \partial_{\xi_1}^l \rho(x,y,\xi) \, d \, \xi_2
\\ = \frac{(-1)^{k_1+k_2}}{[i(x_2-y_2)]^{k_1+k_2} } \int\limits_0^\infty  e^{i \, (x_2-y_2) \, \xi_2 } \, \partial_{\xi_2}^{k_1+k_2} \, \big[(i \xi_2)^{k_1} \,  (-i \xi_2)^{k_2}  \partial_{x_2}^{\ag_2-k_1}  \partial_{y_2}^{\bg_2-k_2} \partial_{x_1}^{\ag_1}  \partial_{y_1}^{\bg_1} \partial_{\xi_1}^l \rho(x,y,\xi)\big]  \, d \, \xi_2.
\end{multline*}
The amplitude in the above integral is in $S^{m-l}(\mO \times \rN^2)$. Therefore,
\begin{equation} |I_{k_1,k_2}(x,y,\xi_1)|  \leq  C \int\limits_0^{M_{x,y} |\xi_1|} \la \xi \ra^{m-l} \, d \, \xi_2 \leq C \la \xi_1 \ra^{m-l+1},\quad \mbox{ for all } (x,y) \in K. 
\end{equation}
This proves (\ref{E:rN}) and finishes the proof.

\end{proof}

\noindent We now consider a generalization of $\mK_\pm$. Namely, let $\ve \in \rN^2$ be a unit vector and let us consider the distributions
\begin{equation} \label{E:mue} \mK_{\pm \ve} (x,y) = \frac{1}{(2 \pi)^2}\intl_{\rN^2} e^{i (x-y) \cdot \xi}  \, a_{\ve} (x,y, \xi) \, \bH(\pm \ve^\perp \cdot \xi) \, d \xi,\end{equation}
and the canonical relation:
$$\mC_\ve= \{(x,s \, \ve; x+ t \, \ve^\perp , s \, \ve): x \in \rN^2, \, t, s \in \rN, \, \ga \neq 0 \}. $$
\begin{prop}\label{P:Rot} We have
\begin{itemize}
\item[a)] For any $n>0$, \barni \mK_{\pm \ve} \in I^m_{\frac{1}{n},0} (\Delta) + I^{m+ \frac{1}{n} - \frac{1}{2}}_{1,\frac{1}{n}}(\mC_\ve).
\earni
\item[b)] Assume that $a_\ve(x,y,\xi)$ vanishes to order $k$ across the line $\ell=\{\xi= r \ve:  r \in \rN\}$. Then, $\mK_{\pm \ve}$ is microlocally in the space $I^{m-k-\frac{1}{2}}(\mC_\ve)$ near $\mC_\ve \setminus \Delta$.
\end{itemize}
\end{prop}
\refprop{P:Rot} can be easily derived from \refprop{P:decompose}, by a rotation argument. We skip it for the sake of brevity.

\medskip

\subsection{Proof of Theorem \ref{T:Main1}} \label{S:Proof}
\noindent For $j=1,2$, let $\rho_j \in C^\infty(\rN^2 \setminus 0)$ be homogeneous of degree zero such that $\rho_j =1$ in a (small) conic neighborhood of $\ell_j \setminus 0$. Moreover, $\rho_j$ is supported inside a small conic neighborhood of $\ell_j$ \footnote{This, in particular, implies $\supp (\rho_j) \cap \supp (\rho_k) =\{0\}$, for $j \neq k$.}. 

\medskip

We can write:
\begin{eqnarray*} 
 \mK(x,y) &=& \frac{1}{(2 \pi)^2}\intl_{\rN^2} e^{i (x-y) \cdot \xi } a(x,y,\xi)\, \chi(\xi) \, d \xi \\ &=&  \frac{1}{(2 \pi)^2}\intl_{\rN^2} e^{i (x-y) \cdot \xi }  \, \Big[1- \sum_{j=1}^2 \rho_j(\xi) \Big] \, a(x,y,\xi)\, \chi(\xi) \, d \xi \\ &+& \sum_{j=1}^2 \frac{1}{(2 \pi)^2}\intl_{\rN^2} e^{i (x-y) \cdot \xi }  \,  \rho_j(\xi) \, a(x,y,\xi)\, \chi(\xi) \, d \xi \\&=& \mK^0(x,y) + \sum_{j=1}^2 \mK^j(x,y).
\end{eqnarray*}
We notice that $\mK^0 \in I^m(\Delta)$ with the amplitude function $$ \Big[1- \sum_{j=1}^2 \rho_j(\xi) \Big] \, a(x,y,\xi)\, \chi(\xi) \in S^m((\rN^2 \times \rN^2) \times \rN^2).$$
Applying the \refprop{P:Rot} for $\mK^1$ and $\mK^2$, we finish the proof.

\section*{Acknowledgements}
The author is thankful to Peter Kuchment for his comments and suggestions. 

\begin{thebibliography}{10}

\bibitem{Am-singular}
{\sc Ambartsoumian, G., Felea, R., Krishnan, V.~P., Nolan, C., and Quinto,
  E.~T.}
\newblock A class of singular {F}ourier integral operators in synthetic
  aperture radar imaging.
\newblock {\em J. Funct. Anal. 264}, 1 (2013), 246--269.

\bibitem{AnUhl}
{\sc Antoniano, J.~L., and Uhlmann, G.~A.}
\newblock A functional calculus for a class of pseudodifferential operators
  with singular symbols.
\newblock In {\em Pseudodifferential operators and applications ({N}otre
  {D}ame, {I}nd., 1984)}, vol.~43 of {\em Proc. Sympos. Pure Math.} Amer. Math.
  Soc., Providence, RI, 1985, pp.~5--16.

\bibitem{Kazantsev}
{\sc Arbuzov, {\`E}.~V., Bukhge{\u\i}m, A.~L., and Kazantsev, S.~G.}
\newblock Two-dimensional tomography problems and the theory of {$A$}-analytic
  functions [translation of {\it {a}lgebra, geometry, analysis and mathematical
  physics ({r}ussian) ({n}ovosibirsk, 1996)}, 6--20, 189, {I}zdat.\ {R}oss.\
  {A}kad.\ {N}auk {S}ibirsk.\ {O}tdel.\ {I}nst.\ {M}at., {N}ovosibirsk, 1997;
  {MR}1624170 (99m:44003)].
\newblock {\em Siberian Adv. Math. 8}, 4 (1998), 1--20.

\bibitem{BFNguyen}
{\sc {Barannyk}, L.~L., {Frikel}, J., and {Nguyen}, L.~V.}
\newblock {On Artifacts in Limited Data Spherical Radon Transform: Curved
  Observation Surface}.
\newblock {\em ArXiv e-prints\/} (June 2015).

\bibitem{Boman}
{\sc Boman, J.}
\newblock An example of nonuniqueness for a generalized {R}adon transform.
\newblock {\em J. Anal. Math. 61\/} (1993), 395--401.

\bibitem{Boman-Quinto}
{\sc Boman, J., and Quinto, E.~T.}
\newblock Support theorems for real-analytic {R}adon transforms.
\newblock {\em Duke Math. J. 55}, 4 (1987), 943--948.

\bibitem{Monvel}
{\sc Boutet~de Monvel, L.}
\newblock Hypoelliptic operators with double characteristics and related
  pseudo-differential operators.
\newblock {\em Comm. Pure Appl. Math. 27\/} (1974), 585--639.

\bibitem{Suresh}
{\sc Eswarathasan, S.}
\newblock Microlocal analysis of scattering data for nested conormal
  potentials.
\newblock {\em J. Funct. Anal. 262}, 5 (2012), 2100--2141.

\bibitem{Fari97}
{\sc Faridani, A., Finch, D.~V., Ritman, E.~L., and Smith, K.~T.}
\newblock Local tomography. {II}.
\newblock {\em SIAM J. Appl. Math. 57}, 4 (1997), 1095--1127.

\bibitem{Fari92}
{\sc Faridani, A., Ritman, E.~L., and Smith, K.~T.}
\newblock Local tomography.
\newblock {\em SIAM J. Appl. Math. 52}, 2 (1992), 459--484.

\bibitem{FeleaCPDE}
{\sc Felea, R.}
\newblock Composition of {F}ourier integral operators with fold and blowdown
  singularities.
\newblock {\em Comm. Partial Differential Equations 30}, 10-12 (2005),
  1717--1740.

\bibitem{FeleaSAR}
{\sc Felea, R., Gaburro, R., and Nolan, C.~J.}
\newblock Microlocal analysis of {SAR} imaging of a dynamic reflectivity
  function.
\newblock {\em SIAM J. Math. Anal. 45}, 5 (2013), 2767--2789.

\bibitem{FeleaQuinto}
{\sc Felea, R., and Quinto, E.~T.}
\newblock The microlocal properties of the local 3-{D} {SPECT} operator.
\newblock {\em SIAM J. Math. Anal. 43}, 3 (2011), 1145--1157.

\bibitem{FLU}
{\sc Finch, D., Lan, I.-R., and Uhlmann, G.}
\newblock Microlocal analysis of the x-ray transform with sources on a curve.
\newblock In {\em Inside out: inverse problems and applications}, vol.~47 of
  {\em Math. Sci. Res. Inst. Publ.} Cambridge Univ. Press, Cambridge, 2003,
  pp.~193--218.

\bibitem{Finch-Range-Attenuated}
{\sc Finch, D.~V.}
\newblock Uniqueness for the attenuated x-ray transform in the physical range.
\newblock {\em Inverse Problems 2}, 2 (1986), 197--203.

\bibitem{Finch-Review}
{\sc Finch, D.~V.}
\newblock The attenuated x-ray transform: recent developments.
\newblock In {\em Inside out: inverse problems and applications}, vol.~47 of
  {\em Math. Sci. Res. Inst. Publ.} Cambridge Univ. Press, Cambridge, 2003,
  pp.~47--66.

\bibitem{FQ13}
{\sc Frikel, J., and Quinto, E.~T.}
\newblock Characterization and reduction of artifacts in limited angle
  tomography.
\newblock {\em Inverse Problems 29}, 12 (2013), 125007.

\bibitem{FQ-Para}
{\sc {Frikel}, J., and {Quinto}, E.~T.}
\newblock {A paradigm for the characterization of artifacts in tomography}.
\newblock {\em ArXiv e-prints\/} (Sept. 2014).

\bibitem{FQ15}
{\sc Frikel, J., and Quinto, E.~T.}
\newblock Artifacts in incomplete data tomography with applications to
  photoacoustic tomography and sonar.
\newblock {\em SIAM Journal on Applied Mathematics 75}, 2 (2015), 703--725.

\bibitem{Frikel-Quinto-2015}
{\sc {Frikel}, J., and {Quinto}, E.~T.}
\newblock {Limited data problems for the generalized Radon transform in
  $\mathbb{R}^n$}.
\newblock {\em ArXiv e-prints\/} (Oct. 2015).

\bibitem{GrUDuke}
{\sc Greenleaf, A., and Uhlmann, G.}
\newblock Nonlocal inversion formulas for the {X}-ray transform.
\newblock {\em Duke Math. J. 58}, 1 (1989), 205--240.

\bibitem{GrU-Functional}
{\sc Greenleaf, A., and Uhlmann, G.}
\newblock Estimates for singular {R}adon transforms and pseudodifferential
  operators with singular symbols.
\newblock {\em J. Funct. Anal. 89}, 1 (1990), 202--232.

\bibitem{GrUCon}
{\sc Greenleaf, A., and Uhlmann, G.}
\newblock Microlocal techniques in integral geometry.
\newblock In {\em Integral geometry and tomography ({A}rcata, {CA}, 1989)},
  vol.~113 of {\em Contemp. Math.} Amer. Math. Soc., Providence, RI, 1990,
  pp.~121--135.

\bibitem{Nov-Kun}
{\sc Guillement, J.-P., Jauberteau, F., Kunyansky, L., Novikov, R., and
  Trebossen, R.}
\newblock On single-photon emission computed tomography imaging based on an
  exact formula for the nonuniform attenuation correction.
\newblock {\em Inverse Problems 18}, 6 (2002), L11--L19.

\bibitem{GU}
{\sc Guillemin, V., and Uhlmann, G.}
\newblock Oscillatory integrals with singular symbols.
\newblock {\em Duke Math. J. 48}, 1 (1981), 251--267.

\bibitem{hormander71fourier}
{\sc H{\"o}rmander, L.}
\newblock Fourier integral operators. {I}.
\newblock {\em Acta Math. 127}, 1-2 (1971), 79--183.

\bibitem{katsevich1998local}
{\sc Katsevich, A.}
\newblock Local tomography with nonsmooth attenuation ii.
\newblock In {\em Inverse Problems, Tomography, and Image Processing}.
  Springer, 1998, pp.~73--86.

\bibitem{Kat-JMAA}
{\sc Katsevich, A.~I.}
\newblock Local tomography for the limited-angle problem.
\newblock {\em J. Math. Anal. Appl. 213}, 1 (1997), 160--182.

\bibitem{RamKat-AML1}
{\sc Katsevich, A.~I., and Ramm, A.~G.}
\newblock Filtered back projection method for inversion of incomplete
  tomographic data.
\newblock {\em Appl. Math. Lett. 5}, 3 (1992), 77--80.

\bibitem{Kuchment-review}
{\sc Kuchment, P.}
\newblock Generalized transforms of {R}adon type and their applications.
\newblock In {\em The {R}adon transform, inverse problems, and tomography},
  vol.~63 of {\em Proc. Sympos. Appl. Math.} Amer. Math. Soc., Providence, RI,
  2006, pp.~67--91.

\bibitem{kuchment2014radon}
{\sc Kuchment, P.}
\newblock {\em The Radon transform and medical imaging}, vol.~85.
\newblock SIAM, 2014.

\bibitem{KLM}
{\sc Kuchment, P., Lancaster, K., and Mogilevskaya, L.}
\newblock On local tomography.
\newblock {\em Inverse Problems 11}, 3 (1995), 571--589.

\bibitem{kun-2001-spect}
{\sc Kunyansky, L.~A.}
\newblock A new spect reconstruction algorithm based on the novikov explicit
  inversion formula.
\newblock {\em Inverse problems 17}, 2 (2001), 293.

\bibitem{Kurusa}
{\sc Kurusa, {\'A}.}
\newblock The invertibility of the {R}adon transform on abstract rotational
  manifolds of real type.
\newblock {\em Math. Scand. 70}, 1 (1992), 112--126.

\bibitem{Markoe-Quinto}
{\sc Markoe, A., and Quinto, E.~T.}
\newblock An elementary proof of local invertibility for generalized and
  attenuated {R}adon transforms.
\newblock {\em SIAM J. Math. Anal. 16}, 5 (1985), 1114--1119.

\bibitem{MU}
{\sc Melrose, R.~B., and Uhlmann, G.~A.}
\newblock Lagrangian intersection and the {C}auchy problem.
\newblock {\em Comm. Pure Appl. Math. 32}, 4 (1979), 483--519.

\bibitem{Natterer-79}
{\sc Natterer, F.}
\newblock On the inversion of the attenuated {R}adon transform.
\newblock {\em Numer. Math. 32}, 4 (1979), 431--438.

\bibitem{Natterer-83}
{\sc Natterer, F.}
\newblock Exploiting the ranges of {R}adon transforms in tomography.
\newblock In {\em Numerical treatment of inverse problems in differential and
  integral equations ({H}eidelberg, 1982)}, vol.~2 of {\em Progr. Sci. Comput.}
  Birkh\"auser, Boston, Mass., 1983, pp.~290--303.

\bibitem{Natterer-01}
{\sc Natterer, F.}
\newblock Inversion of the attenuated {R}adon transform.
\newblock {\em Inverse Problems 17}, 1 (2001), 113--119.

\bibitem{Nat}
{\sc Natterer, F., and Frank, W.}
\newblock {\em Mathematical Methods in Image Reconstruction}, vol.~5.
\newblock SIAM, 2001.

\bibitem{Streak-Artifacts}
{\sc Nguyen, L.~V.}
\newblock How strong are the streak artifacts in limited angle tomography.
\newblock {\em Preprint\/}.

\bibitem{SIMA}
{\sc Nguyen, L.~V.}
\newblock On artifacts in limited data spherical radon transform: Flat
  observation surfaces.
\newblock {\em SIAM Journal on Mathematical Analysis 47}, 4 (2015), 2984--3004.

\bibitem{NoChe}
{\sc Nolan, C.~J., and Cheney, M.}
\newblock Microlocal analysis of synthetic aperture radar imaging.
\newblock {\em J. Fourier Anal. Appl. 10}, 2 (2004), 133--148.

\bibitem{Novikov-CRA}
{\sc Novikov, R.~G.}
\newblock Une formule d'inversion pour la transformation d'un rayonnement {X}
  att\'enu\'e.
\newblock {\em C. R. Acad. Sci. Paris S\'er. I Math. 332}, 12 (2001),
  1059--1063.

\bibitem{Novikov-Attenuated}
{\sc Novikov, R.~G.}
\newblock An inversion formula for the attenuated {X}-ray transformation.
\newblock {\em Ark. Mat. 40}, 1 (2002), 145--167.

\bibitem{Quinto-80}
{\sc Quinto, E.~T.}
\newblock The dependence of the generalized {R}adon transform on defining
  measures.
\newblock {\em Trans. Amer. Math. Soc. 257}, 2 (1980), 331--346.

\bibitem{Quinto-83}
{\sc Quinto, E.~T.}
\newblock The invertibility of rotation invariant {R}adon transforms.
\newblock {\em J. Math. Anal. Appl. 91}, 2 (1983), 510--522.

\bibitem{Quinto-93-pompiu}
{\sc Quinto, E.~T.}
\newblock Pompeiu transforms on geodesic spheres in real analytic manifolds.
\newblock {\em Israel J. Math. 84}, 3 (1993), 353--363.

\bibitem{RamKat-AML2}
{\sc Ramm, A.~G., and Katsevich, A.~I.}
\newblock Inversion of incomplete {R}adon transform.
\newblock {\em Appl. Math. Lett. 5}, 2 (1992), 41--45.

\bibitem{RamKat-Book}
{\sc Ramm, A.~G., and Katsevich, A.~I.}
\newblock {\em The {R}adon transform and local tomography}.
\newblock CRC Press, Boca Raton, FL, 1996.

\bibitem{Shubin-book}
{\sc Shubin, M.~A.}
\newblock {\em Pseudodifferential operators and spectral theory}, second~ed.
\newblock Springer-Verlag, Berlin, 2001.
\newblock Translated from the 1978 Russian original by Stig I. Andersson.

\bibitem{smith1985mathematical}
{\sc Smith, K.~T., and Keinert, F.}
\newblock Mathematical foundations of computed tomography.
\newblock {\em Applied Optics 24}, 23 (1985), 3950--3957.

\bibitem{vainberg1981reconstruction}
{\sc Vainberg, E., Kazak, I., and Kurozaev, V.}
\newblock Reconstruction of the internal three-dimensional structure of objects
  based on real-time integral projections.
\newblock {\em Soviet Journal of Nondestructive Testing 17}, 6 (1981),
  415--423.

\end{thebibliography}

\def\dbar{\leavevmode\hbox to 0pt{\hskip.2ex \accent"16\hss}d}

\end{document}